\documentclass[english, 11pt]{article}
\usepackage[T1]{fontenc}
\usepackage[latin9]{inputenc}
\usepackage{amsmath,amssymb,amsthm}
\usepackage{babel}
\usepackage{fullpage}
\newtheorem{thm}{Theorem}{}{}
\newtheorem{lem}{Lemma}{}{}
{}{}

\newcommand{\beq}{\begin{equation}}
\newcommand{\eeq}{\end{equation}}
\newcommand{\commentout}[1]{}
\newcommand{\be}{\begin{enumerate}}
\newcommand{\ee}{\end{enumerate}}
\newcommand{\RR}{\mathbb{R}}

\usepackage{color}
\usepackage[normalem]{ulem}
\DeclareMathOperator{\conv}{conv}
\DeclareMathOperator{\COR}{COR}
\DeclareMathOperator{\NOR}{NOR}
\DeclareMathOperator{\NOT}{NOT}
\DeclareMathOperator{\poly}{poly}
\DeclareMathOperator{\Tr}{Tr} % Trace

\newcommand{\copositive}{C}
\newcommand{\completelypositive}{C^*}
\newcommand{\PSDcone}{C_\mathrm{PSD}}

\begin{document}

\title{Generalised probabilistic theories and conic extensions of polytopes}

\author{Samuel Fiorini\footnotemark[1]\and
	Serge Massar\footnotemark[2] \and
	Manas K. Patra\footnotemark[2] \and
	Hans Raj Tiwary\footnotemark[3]}

\renewcommand{\thefootnote}{\fnsymbol{footnote}}
\footnotetext[1]{Department of Mathematics, Universit\'e libre de Bruxelles (ULB), Belgium. \texttt{sfiorini@ulb.ac.be}}
\footnotetext[2]{Laboratoire d'Information Quantique, CP225, Department of Physics, Universit\'e libre de Bruxelles (ULB), Belgium. \texttt{\{smassar,manas.kumar.patra\}@ulb.ac.be}}
\footnotetext[3]{Department of Applied Mathematics, Charles University, Prague, Czech Republic. \texttt{hansraj@kam.mff.cuni.cz}}

\date{}

\maketitle

\begin{abstract}
Generalized probabilistic theories (GPT) provide a general framework that includes classical and quantum theories. It is described by a cone $C$ and its dual $C^*$. We show that whether some one-way communication complexity problems can be solved within a GPT is equivalent to the recently introduced cone factorisation of the corresponding communication matrix $M$. We also prove an analogue of Holevo's theorem: when the cone $C$ is contained in $\mathbb{R}^{n}$, the classical capacity of the channel realised by sending GPT states and measuring them is bounded by $\log n$.

Polytopes and optimising functions over polytopes arise in many areas of discrete mathematics. A conic extension of a polytope is the intersection of a cone $C$ with an affine subspace whose projection onto the original space yields the desired polytope. Extensions of polytopes can sometimes be much simpler geometric objects than the polytope itself. The existence of a conic extension of a polytope is equivalent to that of a cone factorisation of the slack matrix of the polytope, on the same cone.

We show that all $0/1$ polytopes whose vertices can be recognized by a polynomial size circuit, which includes as a special case the travelling salesman polytope and many other polytopes from combinatorial optimisation, have small conic extension complexity when the cone is the completely positive cone.

Using recent exponential lower bounds on the linear extension complexity of polytopes, this provides an exponential gap between the communication complexity of GPT based on the completely positive cone and classical communication complexity, and a conjectured exponential gap with quantum communication complexity.

Our work thus relates the communication complexity of generalisations of quantum theory to questions of mainstream interest in the area of combinatorial optimisation.

\end{abstract}

\section{Introduction}

Generalised Probabilitic Theories (GPT)~\cite{Mackey63,Davies70,Edwards71,Foulis81,Ludwig85,Hardy01,Barrett07,Chiribella10,Chiribella11,Masanes11}
are a framework that allows generalisations of both classical and quantum theories.  
In its simplest form a GPT is given by a closed convex cone $C$ that defines the state space, by the dual cone $C^*$ that defines the measurement space, and by a unit $u\in C^*$ that normalises the states.
Upon adding a sufficient set of axioms one restricts to classical or quantum theory. But using only a subset of the axioms provides a framework in which more general theories can be studied. 
Many phenomena considered uniquely quantum, such as no-cloning and no-broadcasting, trade-off between state disturbance and measurement, properties associated with entanglement, teleportation, remote steering of ensembles, and properties of entropy, already appear at the level of GPT, see e.g.~\cite{Klay87,Klay88,Barrett07,Barnum06,Barnum07,Barnum08, Barnum09,Barnum10}.
Related lines of enquiry have shown that non local theories obeying no-signalling have ``quantum'' properties such as intrinsic randomness, impossibility of cloning, secret key generation, see e.g.~\cite{Masanes06,BHK,Acin06}.

 Ideally one would hope to find a set of simple and physically intuitive axioms that naturally restrict GPT to quantum theory~\cite{Fuchs02}. Information and complexity theory provide a possible line of approach by providing criteria that can be used to rule out classes of theories. The development of quantum information~\cite{Nielsen10} shows that perfectly consistent complexity theories alternative to classical  are possible. On the other hand it has been shown that unlimited supply of maximally non local boxes makes communication complexity trivial~\cite{VanDam13,Kaplan11}, which can be taken as an argument for why such correlations are not physical. More recently, the principle of information causality was shown to be violated by many non local correlations~\cite{Pawlowski09}.

Independently of the above, considerable work has been devoted to understanding the geometry and extension complexity of polytopes~\cite{Conforti, Kaibel11, Wolsey11}. For instance, the polytope associated with the Travelling Salesman Problem (TSP)  is the convex hull of all points in $ \{ 0,1 \}^{ {n}\choose{2}}$ that correspond to a Hamiltonian cycle in the complete $n$-vertex graph $K_n$. Solving the TSP is equivalent to linear optimisation over the TSP polytope. Representing the set of feasible solutions of a problem by a polytope forms the basis of a standard
and powerful methodology in combinatorial optimisation, see, e.g., \cite{Schrijver}.

Many polytopes of interest have exponentially many facets, which makes them difficult to use directly. An extension (or lift) of a polytope is a geometric object in a larger dimensional space whose projection onto the original space yields the desired polytope. This is related to the concept of extended formulation, which refers to the description of the extension, here a system of linear equations plus one (conic) constraint. Linear extensions of polytopes are given in terms of linear programs. Semidefinite and conic extensions of polytopes are given in terms of semidefinite programs and conic programs.  The extension may be much simpler than the original polytope, see e.g.~\cite{Conforti}. This motivates the definition of the linear (semidefinite; conic) extension complexity of a polytope as the minimum size of a linear (semidefinite; conic) program expressing the polytope, in terms of the dimension of the cone. When the extension complexity is small, optimisation problems that seem difficult over the original polytope may become simple over the extended formulation.

It was shown in~\cite{Yannakakis,Gouveia} that the existence of a linear (semidefinite; conic) extension  of a polytope is essentially equivalent to certain  linear (PSD; conic) factorisations of a matrix associated to the polytope, called the slack matrix. The slack matrix records for each pair $(v,F)$ of vertex $v$ and facet $F$ of the polytope the corresponding algebraic distance. Specifically the matrix $M$ has a cone factorisation $M = TU$ if $T$ is a matrix whose rows belong to the cone $C$ and $U$ is a matrix whose columns belong to the dual cone $C^*$. When the cone $C$ is the nonnegative orthant or the cone of PSD matrices (positive semidefinite matrices), one obtains the nonnegative and PSD factorisations of the matrix.

As shown in~\cite{Faenza,Fiorini}, the size of a nonnegative (PSD) factorisation of the matrix $M$ is equal, up to a small additive constant, to the number of classical (quantum) bits that must be sent in a randomized one-way communication complexity scenario with nonnegative outputs that computes the matrix in expectation. Conversely, the existence of such a communication complexity protocol implies the corresponding factorisation of the matrix $M$. 
%{\color{red} [Serge: suggest adding these sentences:] 
Note that the communication complexity scenario used here differs from the one most often used in the literature since on the one hand we require that the matrix $M$ be reproduced exactly (we tolerate no error), but on the other hand it must only be reproduced on average (the protocol could for instance output $0$ a large fraction of the time).%}

Inspired by the connection between PSD factorisation and quantum communication complexity, it was shown in~\cite{Fiorini} that the linear extension complexity of some important polytopes from combinatorial optimisation, including the correlation polytope and TSP polytope, is exponential. Extensions and strengthening's of this result can be found in~ \cite{Braun12,Braverman,Braun,Avis,Pokutta,Chan_et_al13}. 

Understanding the semidefinite extension complexity of polytopes is an important research question~\cite{Gouveia,Fiorini,Gouveia12,Gouveia13}. In particular since semidefinite programming is in P, a small semidefinite extension of the TSP polytope with efficiently computable coefficients would imply that P $=$ NP. It is therefore reasonable to conjecture that the semidefinite extension complexity of polytopes such as the TSP polytope is exponential. This conjecture is supported by the counting argument of~\cite{Briet} (based on the earlier work of \cite{RothvoS}) that shows that some $0/1$ polytopes have large semidefinite extension complexity.

In the present work we connect the above two areas of study.
First we give an operational meaning to the cone factorisation of a matrix $M$, for an arbitrary cone $C$: it is equivalent (up to the communication of a single classical bit) to the existence of a randomized one-way communication complexity scenario with nonnegative outputs that computes the matrix $M$ in expectation when the states and measurements are described by the GPT associated to cone $C$. This generalises the operational interpretation of the nonnegative and PSD factorisations in terms of classical and quantum communication complexity.

In order to understand the implications of this result, it is important to have an upper bound on how  much classical information can be stored in a state of a GPT. The analogous result stating that at most one classical bit can be stored in a quantum bit is known as Holevo's theorem~\cite{Holevo}, and underlies much of quantum information theory.
Indeed only in the presence of such a bound can one give meaning to communication complexity of the corresponding GPT.
Our second result is to provide such a bound: namely we show that if a GPT is  associated to a cone $C \subset \mathbb{R}^{n}$, then the states of this GPT can store at most $\log n$ classical bits.
To prove this result, we use the fact  that the space of measurements in a GPT is convex, and then prove that the extremal  measurements have at most $n$ non-zero outcomes. This characterisation of GPT extremal measurements is to our knowledge new and of interest in itself. It generalises a well known characterisation of extremal quantum measurements~\cite{DPP}.

We then consider the specific case of the copositive cone $\copositive_n=\{ X \mid y^\intercal X y \geq 0 , \forall y\in \mathbb{R}^{n}_+\}$ and its dual the completely positive cone $\completelypositive_n=\{ X=\sum_{i=1}^k y_i y_i^\intercal  \mid y_i\in \mathbb{R}^{n}_+\}$ (throughout, the elements of $\mathbb{R}^n$ are column vectors). From the point of view of quantum information, the completely positive cone $\completelypositive_n$ can be viewed as the space of $n\times n$ density matrices that can be expressed as convex combinations of pure states with real nonnegative coefficients (in a specific basis). It is well known that completely positive (copositive) programming, that is, maximising a linear function over the intersection of an affine subspace and the completely positive (copositive) cone, is NP-hard \cite{BDKRT2000}. %This implies that finding polynomial time approximations of the completely positive (copositive) cones is NP-hard, i.e. the difficulty of NP-hard problems underlies the very complicated geometry of the completely positive and copositive cones.
In other words, the very complicated geometry of the completely positive and copositive cones allows one to efficiently encode NP-complete decision problems.

We show here that all polynomially definable $0/1$ polytopes have polynomial size completely positive extension complexity. Such a polytope is a polytope whose vertices form a subset of $\{0,1\}^d$ that can be recognized by a circuit of size $\poly(d)$. To prove this result we proceed in two steps. First we show, extending the work of Maximenko \cite{Maximenko12}, that all polynomially definable $0/1$ polytopes are projections of faces of the correlation polytope $\COR(n)= \conv\{aa^\intercal  \mid a \in \{0,1\}^n\}$, with $n =\poly(d)$. Second, exploiting results of Burer~\cite{Burer09}, we show that the correlation polytope has a polynomial size completely positive extension, that is $\COR(n)$ is given by the projection of the intersection of $\completelypositive_{\poly(n)}$ with an affine subspace. This result is interesting by itself because small completely positive (or copositive) programming formulations have been found for a large number of combinatorial optimisation problems, see, e.g.~\cite{QKRT98,BDKRT2000,KP02,Burer09}. We show that virtually all combinatorial optimisation problems share this property. In fact, Burer~\cite{Burer09} asks: ``Other than the handful of problems listed above, what types of problems can be represented as COPs [copositive programs] or as CPPs [completeley positive programs]?''. Our result can be viewed as an answer to this question: all combinatorial optimisation problems (integer linear programs with 0/1 variables) such that testing the feasibility of a solution can be done efficiently, can be efficiently represented as completely positive programs.

Using the correspondence outlined above, this result implies an exponential gap between the communication complexity of GPT based on the completely positive cone and classical communication complexity. In view of the very plausible conjecture mentioned above, one also anticipates an exponential gap with quantum communication. 

We now return to the problem of introducing sets of axioms that reduce to quantum theory. In general there will be an interplay between structural axioms that define the mathematical framework (e.g. that states are vectors in $\mathbb{R}^{n}$), physical axioms (e.g. that convex combinations of states are states), and information theoretic axioms (e.g. that certain communication complexity tasks are impossible, or that secure key distribution is possible). In particular the present work suggests that even at the very basic level of GPT, complexity arguments could be used to rule out certain theories. Indeed our results show that GPT based on the completely positive cone $\completelypositive_n$ provides exponential saving over classical (conjectured quantum) communication, and this could be used to rule out this theory. (There are probably many other reasons to rule out GPT based on $\completelypositive_n$, but these would invoke other axioms, related for instance to transformations between states).
The present work can thus be viewed as a step along the program of \cite{Fuchs02,Brassard05} who wish to use as much as possible information theory type axioms to restrict possible physical theories.

As a concluding remark, we note that there have already been a number of results in classical complexity that were obtained through quantum arguments, or inspired by quantum information, see e.g.~\cite{Kerenidis,Aaronson,Fiorini} and the review~\cite{Drucker}. Here the same kind of connection occurs, but with ideas and arguments inspired by the foundations of quantum mechanics, and in particular generalisations of  quantum theory. The connection arose very naturally during the development of the present work: we first realised that the recently introduced cone factorisation of matrices could be given an operational interpretation within the context of GPT, and then explored to what exent the completely postive cone would provide an interesting example, which finally lead to new results in combinatorial optimisation.

The reader mainly interested in the foundation of physics aspects should concentrate on Sections 2 to 5. On the other hand, the reader interested in the combinatorial optimisation aspects should go to the self contained Section 6.

\section{Generalised Probabilistic Theories}

\subsection{General formulation}

We work in $\mathbb{R}^{n}$ with the usual scalar product which we
denote $\langle\cdot,\cdot\rangle$. Let $C\subset\mathbb{R}^{n}$ be a proper cone (i.e. $C$ is a closed, pointed and full-dimensional cone), and denote by $C^{*} = \{x \in \mathbb{R}^n \mid \forall y \in C : \langle x, y \rangle \geq 0\}$ its dual. Notice that $C^*$ is again a proper cone.

An element $\omega\in C$ is an \emph{unnormalised state}. An element $e\in C^{*}$ is an \emph{unnormalised effect}. The \emph{unit effect} $u\in C^{*}$ is an interior point of the dual cone. Thus $\langle u,\omega \rangle > 0$ for every non-zero $\omega\in C$. %The unit effect defines a norm and hence a topology on the space. 

\emph{Normalised states} are states $\omega\in C$ such that $\langle\omega,u\rangle=1$. 
Any unnormalised state $\omega\neq0$ can be rescaled $\omega\to \omega / \langle\omega,u\rangle$ to become a normalised state.
Normalised states form a closed convex set. Convex combinations
of states  correspond to probabilistic mixture: for $0 \leq p \leq 1$,
$p\omega_{1}+(1-p)\omega_{2}$ can be interpreted as the state obtained by preparing $\omega_{1}$ with probability $p$ and preparing $\omega_{2}$ with probability $1-p$.

A \emph{measurement} $M$ is a finite set of effects that sum to the unit effect:
$M=\left\{ e_{i}\in C^{*} \mid \sum_{i}e_{i}=u\right\} $.
Note that any effect $e\in C^{*}$ , $e\neq0$ can be rescaled so that $\{ \lambda e,  u - \lambda e\}$ is a measurement since $u$ is an interior point. 

The above construction allows one to study one-way communication scenarios as follows.
One party, Alice, prepares
a (normalised) state $\omega$ and sends it to another party, Bob, who
carries out a measurement $M$ on the state. The probability that
Bob obtains outcome $i$ is 
\[
P(i|\omega)=\langle e_{i},\omega\rangle\ .
\]
These are indeed probabilities since from the definitions, $P(i|\omega)\geq 0$
and $\sum_{i}P(i|\omega)=1$.

A \emph{Generalised Probabilistic Theory}, in the simple form used here, is therefore defined using the above construction by a proper cone $C\subset\mathbb{R}^{n}$ and a unit $u \in C^*$. We denote it GPT($C$,$u$).

\subsection{Classical theory. }

Classical theory corresponds to the case where the cone $C=\left\{ x\in\mathbb{R}^{n} \mid \forall i : x_{i}\geq0\right\} $
is the nonnegative orthant and the unit effect is $u=\left(1,1,\ldots,1\right)^\intercal$.
Then the normalised states belong to the simplex $\Delta_n = \{x \in \mathbb{R}^n \mid \forall i : x_{i} \geq 0, \sum_{i}x_{i}=1\}$.
This simplex has $n$ extreme points $\omega_{i}=\left(0,0,\dotsc ,1,0,\dotsc ,0\right)^\intercal$.
Any normalised state $\omega$ has a unique decomposition as a convex
combination of the extreme points $\omega=\sum_{i}p_{i}\omega_{i}$
with $p_{i}\geq0$ for all $i$ and $\sum_{i}p_{i}=1$. One can therefore view $\omega$ as a probability distribution over the extreme points.

The dual cone $C^{*}$ is also the nonnegative orthant. There is a canonical
measurement with effects $e_{i}=\left(0,0,\ldots,1,0,\ldots,0\right)^\intercal$, $i=1,\ldots,n$. The probability that one gets result $i$ in state $\omega$ is $p_{i}$, i.e. the probability that the system was in extreme point $\omega_{i}$.

\subsection{Quantum theory}

Quantum theory corresponds to the case where the cone $\PSDcone=\PSDcone^{*}$ is the set of positive semidefinite hermitian matrices. If the Hilbert space dimension is $d$ (over the complex numbers), then $\PSDcone$ is a proper full-dimensional cone in the space of all $d \times d$ matrices (this time over the reals). Thus $n = d + 2 {d \choose 2} = d^2$ here. The scalar product can be written as $\langle\omega,e\rangle = \Tr \left(\omega e\right)$. The unit effect is the identity matrix $u=I$. A state $\omega\in \PSDcone$ is normalised if $\Tr(\omega)=1$. The extreme states are called \emph{pure states}. They correspond to rank 1 positive semidefinite matrices with unit trace. A measurement $M=\left\{ e_{i}\in \PSDcone \mid \sum_{i}e_{i}=I\right\}$ is called a Positive Operator Valued Measure (POVM).

\subsection{GPT based on the completely positive and copositive cones}

Let $\mathbb{S}^d$ denote the set of all $d \times d$ real symmetric matrices. The \emph{cone of completely positive matrices} is the set of matrices
\begin{equation}
\completelypositive_d=\left\{X \in \mathbb{S}^d \mid X=\sum_k z_k z_k^\intercal \mbox{ for a finite set } \{z_k\in\mathbb{R}^{d}_+\}
\right\} \ .
\end{equation}
It can be thought of as the restriction of quantum theory to states with real nonnegative coefficients (in a preferred basis), since any matrix in $\completelypositive_d$ is the convex combination of pure states with nonnegative real coefficients.

Its dual (relative to the scalar product $\langle X, Y\rangle = \Tr(XY)$) is the \emph{cone of copositive matrices}
\begin{equation}
\copositive_d =\left\{X \in \mathbb{S}^{d} \mid  z^\intercal  X z \geq 0 \mbox{ for all } z\in\mathbb{R}^{d}_+
\right\}\ .
\end{equation}

We can take the unit effect to be the unit matrix $u=I=\sum_{i=1}^n e_i e_i^\intercal$. The normalised pure states are of the form $X=(X_{ij}) = (p_i p_j) = pp^\intercal$ with $\sum_i p_i^2=1$.
Alternatively we could take the unit effect to be
$u=J$, the matrix with all entries $1$. Let $f= (1,\dotsc,1)^\intercal$. Since $J=f f^\intercal$ it is also an (unnormalised) state. In this case the normalised pure states are of the form $X=(X_{ij}) = (p_i p_j) = pp^\intercal$ with $p = (p_1,\dotsc,p_d)^\intercal$ a probability distribution. Note that we have $\completelypositive_d\subset \copositive_d$, i.e. the state space is strictly smaller than the effect space.

The dual of the copositive cone is the completely positive cone (as implied by the traditional notation used above). Hence we could also consider the dual theory where $\copositive_d$ constitutes the state space and $\completelypositive_d$ the set of measurements. The unit effect could be taken to be $I$ or $J$ as above. In this case the effect space is smaller than the state space. (However it is the case where the set of states is $\completelypositive_d$ that will interest us here).

\section{Extremal measurements}

\subsection{Refining measurements}

A measurement $M'$ is a {\em refinement} of measurement $M$ if 
\begin{eqnarray}
M&=&\left\{ e_{i}\in C^{*} \mid i=1,\ldots,m,\ \sum_{i}e_{i}=u\right\} \nonumber\\
M'&=&\left\{ f_{i,j}\in C^{*} \mid i=1,\ldots,m,\ j=1,\ldots,\ell_i,\ \sum_{i,j}f_{i,j}=u\right\}\nonumber\\
&\mbox{and}&
\forall i = 1, \ldots, m : e_i=\sum_j f_{i,j}\ .
\end{eqnarray}
That is, $M'$ is a refinement of $M$ if  carrying out measurement $M$ is equivalent to carrying out measurement $M'$ and then forgetting part of the information contained in the outcome (in the notation above, forgetting the label $j$ of the outcome $i,j$, and keeping only label $i$). Carrying out the refined measurement $M'$ will in general provide more information than carrying out the original measurement $M$.

An {\em extremal} vector $v$ in a proper cone $C$ is defined by the property that for any $w\in C$, $w \leq_C v$ implies $w = \lambda v$ for some $\lambda \geq 0$, where we use the notion of generalized inequality: if $x,y \in C$, $x \leq_C y$ iff $y-x \in C$. Equivalently, vector $v \in C$ is extremal if and only if $v = w + z$ with $w, z \in C$ implies $w = \lambda v$ for some $\lambda \geq 0$. 

\begin{lem}\label{lem:Ref1}
Any measurement $M$ can be refined to a measurement $M'$ whose effects are all extremal vectors of $C^*$.
\end{lem}

This follows immediately from the Krein-Milman theorem~\cite{Barvinok} adapted to cones that states that any vector in a cone $C\subset\mathbb{R}^n$ can be written as a nonnegative combination of extremal vectors. Moreover, by Caratheodory's theorem~\cite{Barvinok}, $v$ can in fact be written as a nonnegative combination of at most $n$ extremal vectors. 

\subsection{Convex combinations of measurements}

A measurement $M$ is \emph{a convex combination} of the measurements $M_1$ and $M_2$ if all the measurements have the same number of outcomes, and if the effects of measurement $M$ are convex combinations with fixed weights of the effects of measurements $M_1$ and $M_2$.
More precisely, if
\begin{eqnarray}
M&=&\left\{ e_{i}\in C^{*} \mid i=1,...,m\  ,\ \sum_{i}e_{i}=u\right\} \nonumber\\
M_1&=&\left\{ f_{i}\in C^{*} \mid i=1,...,m\ ,\ \sum_{i}f_{i}=u\right\} \nonumber\\
M_2&=&\left\{ g_{i}\in C^{*} \mid i=1,...,m\ ,\ \sum_{i}g_{i}=u\right\} \nonumber\
\end{eqnarray}
then we write $M= p M_1 + (1-p) M_2$, $0< p<1$, if for all $i$
\begin{equation}
e_i=p f_i + (1-p)g_i \ .
\end{equation}
Note that some of the $e_i, f_i, g_i$ may be equal to zero.

Operationally this means that the measurement $M$ can be realized by carrying out measurement $M_1$ with probability $p$ and measurement $M_2$ with probability $1-p$, and then keeping only the label of the outcome, but forgetting which of the two measurements was in fact realized.  Carrying out measurement $M_1$ with probability $p$ and measurement $M_2$ with probability $1-p$ will in general provide more information than carrying out the original measurement $M$.

\subsection{Extremal measurements}

\begin{thm}\label{thm:Measure1}
Any measurement $M = \left\{ e_{i}\in C^{*} \mid i = 1, \ldots, m,\ \sum_{i}e_{i}=u\right\}$ can be refined to a measurement $M'$ such that $M'$ can be written as the convex combination of measurements each of which has at most $n$ nonzero effects which are all extremal vectors of $C^*$.
\end{thm}

The analogue of this result for quantum measurements is well  known. In a Hilbert space of dimension $d$, the extremal POVM's with more than $d^2$ elements have only $d^2$ nonzero elements, and these elements are rank one projectors~\cite{DPP}. Note that if one fixes the number of outcomes of the POVM (rather than first carrying out refinement as above), then the structure of extremal POVM's is more complicated, see~\cite{DPP}.

\begin{proof}[Proof of Theorem \ref{thm:Measure1}.] 
Given any measurement $M$, we use Lemma \ref{lem:Ref1} to construct a refined measurement, all of whose effects are extremal.  
Hence from now on, we suppose that the measurement $M$ is composed of $m$ extremal effects. If the number of effects $m \leq n$, then the assertion is trivial. So assume $m>n$, and that all $e_i \neq 0$.

By Carath{\'e}odory's Theorem~\cite{Barvinok}, there are coefficients $\lambda_i \geq 0$ such that $\sum_{i} \lambda_{i} e_{i} = u$ and at most $n$ of the $\lambda_i$'s are nonzero. We denote
$$\lambda_{\mathrm{max}} = \max_i \lambda_{i}\ .$$

First we observe that $\lambda_{\mathrm{max}}>1$. To prove this, suppose that $\lambda_{\mathrm{max}}\leq 1$. Then using the notion of generalized inequality and the fact that $M$ has $m>n$ nonzero effects we have $u = \sum_i \lambda_i e_i \leq_{C^*} \sum_{i : \lambda_i \neq 0} e_{i} <_{C^*} \sum_{i} e_{i}=u$, a contradiction. We define the measurement $M_1=\left\{ f_{i}\in C^{*} \mid i=1,\ldots,m \right\}$ with effects $f_{i} =\lambda_{i} e_{i}$ if $\lambda_i \neq 0$, and $f_{i}=0$ otherwise.
Note that measurement $M_1$ has at most $n$ nonzero effects.

We define measurement $M_2=\left\{ g_{i}\in C^{*} \mid i=1,\ldots,m \right\} $ with effects $g_i = (1-\frac{1}{\lambda_{\mathrm{max}}})^{-1}\left(e_i - \frac{f_i}{\lambda_{\mathrm{max}}}\right)$.
We note that $M_2$ is a legitimate measurement, since
$g_i\in C^*$ (use $\lambda_{\mathrm{max}}>1$) and $\sum_i g_i=u$.
We note that measurement $M_2$ has at most $m-1$ nonzero effects. Indeed for all $i$ such that 
$ \lambda_{i} = \lambda_{\mathrm{max}}$ (by definition of $\lambda_{\mathrm{max}}$ there is at least one such $i$), we have $g_{i}=0$.

Finally, measurement $M$ is the convex combination of measurements $M_1$ and $M_2$ with weights
$$M= \frac{1}{\lambda_{\mathrm{max}}} M_1 + \left(1-\frac{1}{\lambda_{\mathrm{max}}}\right) M_2$$
(These weights are nonnegative since $\lambda_{\mathrm{max}}>1$).
Note that from the above construction all the effects of $M_1$ and of $M_2$ are proportional to the effects in $M$, hence the effects of $M_1$ and $M_2$ are  extremal vectors of $C^*$.

We have thus written measurement $M$ as a convex combination of two measurements, one of which has $n$ nonzero outcomes, and the other at most $m-1$ nonzero outcomes, both of which consist only of extremal vectors of $C^*$. By iterating the argument, $M$ can be written as a convex combination of measurements with at most $n$ nonzero effects, all of which are extremal vectors of $C^*$.
\end{proof}

\section{Holevo bound for GPT}

How much classical information can be stored or transmitted using states $\omega\in C$ of a generalized probabilistic theory?   The corresponding  result in quantum information states that at most one classical bit can be stored in a quantum bit. This is known as Holevo's theorem~\cite{Holevo} and it underlies much of quantum information theory. For instance, for communication complexity problems it is the benchmark that allows meaningful comparison between sending classical and quantum information \cite{Nielsen10}.

To answer this question we consider the following scenario, formulated using a GPT($C,u$): Alice receives some classical input $x$, distributed according to some probability distribution $p(x)$. She encodes it into a state $\omega(x)\in C$ which she sends to Bob. Bob carries out a measurement $M$, obtaining outcome $y$. The classical capacity of the channel is the mutual entropy $I(X;Y)$ between $x$ and $y$, maximized over the probability distribution $p(x)$, the coding $\omega(x)$ and the measurement $M$. 

The {\em Holevo capacity} of a noiseless channel defined by GPT($C,u$) is
 \[ I_H(C,u) = \max_{M,\omega(x)} I(X;Y) \]
 Operationally, it corresponds to choosing an encoding and measurement that maximises the classical capacity of the channel.

\begin{thm}\label{thm:capacity}
For a noiseless channel using GPT($C,u$) states in dimension $n$ (i.e. $C\subset \mathbb{R}^n$), the Holevo capacity is bounded by
 \[ I_H(C,u) \leq \log{n} \]
\end{thm}

\begin{proof}
Denote by $I(X;Y\vert \{\omega(x)\}, M)$ the capacity of the channel for  encoding $\omega(x)$  and measurement $M$. 
It follows from the data processing inequality for mutual information \cite{Cover-Thomas} that if $M'$ is a refinement of $M$, then 
$I(X;Y\vert \{\omega(x)\}, M)\leq I(X;Y\vert \{\omega(x)\}, M')$.
It also follows from convexity of mutual information \cite{Cover-Thomas} that 
if measurement $M= p M_1 + (1-p) M_2$ is the convex combination of measurements $M_1$ and $M_2$, then 
$I(X;Y\vert M) \leq p I(X;Y\vert M_1) + (1-p) I(X;Y\vert M_2)$.
Finally note that if  measurement $M$ has $k$ nonzero effects, then $I(X;Y\vert M) = H(Y) - H(Y|X) \leq \log k$. The result then follows from Theorem \ref{thm:Measure1} and the fact that the encoding and measurement can be arbitrarily chosen. 
\end{proof}

\section{Randomized one-way communication complexity with nonnegative outputs and cone factorisation of matrices}\label{1wayCC}

Consider the following communication complexity problem. Two parties, Alice and Bob, have to implement randomized computation of a function $f(x,y)$ through communication given inputs $x$ for Alice and $y$ for Bob. The expectation value of the random function $f$, is required to satisfy $E(f(x,y))= C_{xy}$, where $(C_{xy})$ is the communication matrix (not to be confused with the cone $C$). Assuming that Alice and Bob share a noiseless channel for states in some GPT($C,u$), we can frame the problem as follows. Alice receives input $x\in\left\{ 0,1\right\} ^{k}$ and Bob receives input $y\in\left\{ 0,1\right\} ^{\ell}$. Upon receiving her input, Alice prepares a normalised state $\omega(x)\in C$ that she sends to Bob. Bob carries out a measurement $M(y)=\left\{ e_{i}(y)\in C^*\right\}$ on the state prepared by Alice. He obtains result $i$ with probability $P(i)=\langle \omega(x),e_{i}(y)\rangle$. Bob then produces an output $r(i,y)$ that depends on the result $i$ of his measurement and on $y$.
We require  that the result output by Bob is always nonnegative: $r(i,y)\geq 0$. We further require that the expectation of the outputs, 
\begin{equation}
E(r\vert xy)=\sum_{i}r(i,y)\langle \omega(x),e_{i}(y)\rangle = C_{xy}\ ,
\label{Exy}
\end{equation}
where $C_{xy}\geq 0$ is the communication matrix.

Let us relax the constraints on Alice slightly, and allow Alice to send to Bob subnormalised state $\langle u,\omega(x)\rangle \leq 1$. Physically, this can be done by providing Alice with an extra classical communication channel of capacity $1$ bit. Alice then sends the state $\omega(x)$ to Bob, and uses the extra bit to tell Bob whether he must output $0$, or carry out the procedure outlined above. 

\begin{thm}\label{thm:conefact} Consider the one-way communication complexity problem described above based on a GPT($C,u$), in which Alice sends possibly subnormalised states $\omega(x)\in C$, Bob carries out measurements $M(y)=\left\{ e_{i}(y)\in C^*\right\} $, and Bob's outputs $r(i,y)\geq 0$ are nonnegative. Then Alice and Bob can produce as expected output the communication matrix $C_{xy} \geq 0$ if and only if there is a \emph{cone factorisation} of the communication matrix on cones $C$ and $C^*$:
\begin{equation}
C_{xy}=\langle \tilde\omega(x),r(y)\rangle
\end{equation}
with $\tilde\omega(x)\in C$, $r(y)\in C^{*}$.
\end{thm}

This result generalises the link between nonnegative factorisation of a matrix and classical communication complexity problem~\cite{Yannakakis,Faenza}, and between PSD factorisation and quantum communication complexity~\cite{Fiorini}, to arbitrary cones and GPT communication complexity. Note that this result is independent of the choice of unit $u\in C^*$, as long as $u$ is an interior point of $C^*$.

\begin{proof}[Proof of Theorem \ref{thm:conefact}]
Let us denote by $b \in\{0,1\}$ the extra bit sent by Alice, and by $p(b\vert x)$ the probability that the bit is $0$ or $1$.
The average output given inputs $x,y$ is then
$E(r\vert xy)=p(b=1\vert x) \sum_{i}r(i,y)\langle \omega(x),e_{i}(y)\rangle$.
Let us define $\tilde \omega(x)=p(b=1\vert x) \omega(x) \in C$ and $r(y)=\sum_{i}r(i,y)e_{i}(y)\in C^*$ . Then producing the communication matrix $C_{x,y}$ implies that we can write $C_{xy}=\langle \tilde\omega(x),r(y)\rangle$.

Conversely, suppose we can write $C_{xy}=\langle \tilde\omega(x),r(y)\rangle$ with $\tilde\omega(x)\in C$, $r(y)\in C^{*}$. Then there exists  $\lambda>0$ such that we can write $C_{xy}=\langle \omega(x),\lambda r(y)\rangle$ with $\lambda>0$ and $\omega(x)=\tilde\omega(x)/\lambda$ such that $\langle u,\omega(x)\rangle\leq1$. We identify $\omega(x)$ as the subnormalised states sent by Alice on input $x$.

Furthermore, since $u$ is an interior point of $C^*$, there exists $\mu>0$ such that $\mu u \geq_{C^*} \lambda r(y)$ for all $y$ (where we use the notion of generalised inequality). This is equivalent to stating that the two vectors $e_1(y)=\lambda r(y)/\mu$ and $e_0(y)=u - \lambda r(y)/\mu$ both belong to $C^*$ for all $y$.  We then consider the measurement $M(y)=\left\{ e_{0}(y), e_1(y)\right\} $ and outputs $r(0,y)=0$, $r(1,y)= \mu/\lambda$.
We  write
$C_{xy}=\sum_{i=0,1} r(i,y) \langle \omega(x),e_i(y) \rangle$, which is of the form eq. (\ref{Exy}), thereby proving the assertion.
\end{proof}

In the next section we discuss one-way communication complexity when the states belong to the completely positive cone, and the effects to the copositive cone. We exhibit a communication matrix, specifically, the slack matrix of the correlation polytope $\COR(n)$, that can be realised by sending states in $\completelypositive_d$ with $\log d=O(\log n)$. On the other hand this problem requires $\Omega(n)$ classical bits of communication. It is highly plausible that it also cannot be achieved using a logarithmic number of quantum bits of communication (the contrary would come close to proving that P $=$ NP, as we discuss below). 

\section{Polytopes}

\subsection{Conic extensions of polytopes}

A polytope $P \subseteq \mathbb{R}^d$ can be described either as the convex hull of a finite set of points $V = \{v_1, \ldots, v_m\} \subseteq \mathbb{R}^d$ or as the set of solutions of a finite system of linear inequalities $Ax \geq b$, where $A\in \mathbb{R}^{n \times d}$, $b\in \mathbb{R}^{n}$, provided that this set of solutions is bounded (see \cite{Ziegler95} for a thorough treatment of the subject). Thus $P$ has the following inner and outer descriptions:
$$
P = \conv(V) = \{x\in \mathbb{R}^d \mid A x \geq b\}.
$$

The slack matrix $S \in \mathbb{R}^{m\times n}$ of the polytope with respect to the above descriptions is the matrix obtained by computing by how much each vertex satisfies each inequality, i.e., it is given by $S_{ij}=A_{j} v_{i}-b_j$, where $A_j$ is the $j$th row of $A$. By definition all elements of a slack matrix are nonnegative. (We remark that, compared to some previous work, in particular~\cite{Fiorini}, we work here with a transposed slack matrix here ---this turns out to be more natural in the present context.)

Let $C \subseteq \RR^k$ be an arbitrary closed convex cone and $C^*$ its dual cone. A \emph{conic extension} of the polytope $P$ is a set 
$Q=\{(x,y)\in\RR^{d+k} \mid Ex+Fy=g,\ y\in C\}$  where $E\in\RR^{p\times d},F\in\RR^{p\times k},$ and $g\in\RR^p$ such that the projection of $Q$ into the $x$-space equals $P$:
$$
\{x\in\RR^d\mid\exists y \in \RR^k : (x,y)\in Q\}=P.
$$
The extension is called \emph{proper} if the affine subspace defined by $Ex+Fy=g$ contains some interior point of $C$.

Given cone $C$, the existence of a conic extension of the polytope $P$ is essentially equivalent to the existence of a cone factorisation of the slack matrix of $P$ \cite{Gouveia}, in the following sense. 

\begin{thm}[Gouveia, Parrilo and Thomas~\cite{Gouveia}] \label{thm:GPT}
Let $P$ be a polytope that is neither empty or a point, and $S$ be any slack matrix of $P$.

\begin{itemize}
\item If $P$ admits a proper conic extension with respect to cone $C$, then its slack matrix $S$ admits a cone factorisation on cones $C$ and $C^*$. 
\item Conversely, if $S$ admits a cone factorisation on $C$ and $C^*$ then $P$ admits a (non-necessarily proper) cone extension with respect to $C$.
\end{itemize}
\end{thm}

Gouveia \emph{et al.}~\cite{Gouveia} show that the technical condition of being proper can be removed if cone $C$ is \emph{nice}, that is, $K^* + F^\perp$ is closed for all faces $F$ of $K$. For instance, it is known that the nonnegative orthant and the PSD cone are both nice.

From Theorem~\ref{thm:conefact}, it follows that it the existence of a conic extension with respect to $C$ is also (essentially) equivalent to the existence of a one-way communication protocol using states belonging to cone $C$ and effects belonging to the dual cone $C^*$, that produces as average output the slack matrix.

\subsection{Correlation polytope}

The \emph{correlation polytope}  $\COR(n)$ is defined as the convex hull of all the rank-$1$ binary symmetric matrices of size $n \times n$. In other words, 
\[
\COR(n) := \conv \{aa^{\intercal} \in \mathbb{R}^{n \times n} \mid a \in \{0,1\}^n\}.
\]

It is shown in~\cite{Fiorini} that any linear extension of the correlation polytope has size $2^{\Omega(n)}$, that is, there exists a constant $\alpha > 0$ such that any linear extension of the correlation polytope has size at least $2^{\alpha n}$. By \emph{linear extension} we mean that the cone $C$ is taken to be the nonnegative orthant $\RR^k_+$. This implies that any one-way classical communication protocol with nonnegative outputs that produces the slack matrix of $\COR(n)$ in expectation requires at least $\alpha n$ classical bits, i.e., the dimension of the space in which the classical information is coded is at least $2^{\alpha n}$.

It is reasonable to conjecture that there does not exist a semidefinite extension of the correlation polytope of size $\poly(n)$. By \emph{semidefinite extension} of size $k$ we mean that the cone $C$ is taken to be the cone of $k \times k$ (real) PSD matrices. Indeed if such a polynomial size semidefinite extension exists, and if an approximation to the coefficients defining this PSD extension could be computed in polynomial time, then there would exist a polynomial time algorithm for maximisation of a linear function over the correlation polytope (since semidefinite programming is in P). But maximising a linear function over $\COR(n)$ is NP-complete. Hence this would imply that P $=$ NP. Just the existence of a $\poly(n)$-size semidefinite extension of $\COR(n)$ would imply NP $\subseteq$ P/poly, as follows from the results of Br\"iet, Dadush and Pokutta~\cite{Briet}.

If this conjecture is true, and $\COR(n)$ does not have a $\poly(n)$-size PSD extension, then there are no quantum one-way communication protocols using quantum states belonging to a Hilbert space of size $\poly(n)$ with nonnegative outputs that produce the slack matrix of $\COR(n)$ in expectation. This follows from the relation between PSD extensions, PSD factorisation, and quantum communication given in \cite{Fiorini}

On the contrary, if one places oneself in the context of a GPT wherein states belong to the completely positive cone and effects to the copositive  cone, then there exists a one-way communication protocol with nonnegative outputs that produces in expectation the slack matrix of the correlation polytope, and in which the information is coded in a space of dimension $\poly(n)$. Hence this provides an exponential saving with respect to classical communication, and a conjectured super-polynomial saving with respect to quantum communication. Similarly to semidefinite extension, we say that the \emph{size} of a completely positive extension is $k$ if it is relative to the cone of all $k \times k$ completely positive matrices.

\begin{thm} \label{thm:COR_poly}
There exists a polynomial size completely positive extension of the correlation polytope.
\end{thm}

\begin{proof} 
Consider an optimization problem of the form
$$
\begin{array}{rll}
\min & x^\intercal Qx+2c^\intercal x&\\
\text{s.t.} & a_i^\intercal x=b_i &\forall i=1,\ldots,m\\
& x_j \geqslant 0 &\forall j=1,\ldots,n\\
& x_j\in\{0,1\}~&\forall j\in B
\end{array}
$$

Suppose that for every $x$ that satisfies $x_j \geqslant 0$ for all $j = 1, \ldots, n$ and $a_i^\intercal x=b_i$ for all $i=1,\ldots,m$ we have that $x_j \leqslant 1$ for all $j \in B$. Burer \cite{Burer09} has shown that, under this assumption, the above problem can be rewritten as the following conic program:
$$
\begin{array}{rll}
\min & Q \bullet X + 2c^\intercal x&\\
\text{s.t.}& a_i^\intercal x=b_i&~\forall i=1,\ldots,m\\
& a_i^\intercal X a_i=b_i^2& ~\forall i=1,\ldots,m\\
& x_j=X_{jj}&~\forall j\in B\\
&\left(\begin{array}{cc}1&x^\intercal\\x&X\end{array}\right)\in \completelypositive_{1+n}
\end{array}
$$
where $\bullet$ denotes the Frobenius product (so $Q \bullet X = \Tr(QX) = \langle Q, X \rangle$) and $C^*_{1+n}$ denotes the completely positive cone generated by the matrices $zz^\intercal$, where $z \in \RR^{1+n}_{\geqslant 0}$.

\medskip
Now, optimizing over the correlation polytope $\COR(n)$ can be modeled as the following optimization problem:
$$
\begin{array}{rll}
\min & x^\intercal Q x\\
\text{s.t.} & x_j\in\{0,1\}~&\forall j=1,\ldots,n
\end{array}
$$
where $Q = (q_{ij}) \in \RR^{n \times n}$. After doubling the size of the vector $x$, this can be rewritten in slack form as
$$
\begin{array}{rll}
\min & x^\intercal 
\begin{pmatrix}
Q & 0\\
0 & 0
\end{pmatrix}
x\\
\text{s.t.} & (e_i+e_{i+n})^\intercal x =1~&\forall i=1,\ldots,n\\
& x_j \geqslant 0 &\forall j=1,\ldots,2n\\\
& x_j\in\{0,1\}~&\forall j=1,\ldots,n,
\end{array}
$$
where $e_i$ denotes the $i$th unit vector in $\RR^{2n}$. It is easy to see that for every $x$ that satisfies $(e_i+e_{i+n})^\intercal x =1$ for all $i=1,\ldots,n$ and $x_j \geqslant 0$ for all $j=1,\ldots,2n$, we have that $0 \leqslant x_j \leqslant 1$ for all $j = 1, \ldots, n$ since these conditions are just a rephrasing of $x_i+x_{i+n}=1$, $x_i\geqslant 0$, $x_{i+n}\geqslant 0$ for $i=1,\ldots,n$. Thus the above formulation satisfies the condition for Burer's result. Therefore, we obtain an equivalent optimization program over the copositive cone:
$$
\begin{array}{rll}
\min & \begin{pmatrix}
Q & 0\\
0 & 0
\end{pmatrix} \bullet X&\\
\text{s.t.} &(e_i+e_{i+n})^\intercal x=1&~\forall i=1,\ldots,n\\
& (e_i+e_{i+n})^\intercal X (e_i+e_{i+n})=1&~\forall i=1,\ldots,n\\
& x_j=X_{jj}&~\forall j=1,\ldots,n\\
&\left(\begin{array}{cc}1&x^\intercal\\x&X\end{array}\right)\in C^*_{1+n+n}.
\end{array}
$$
Finally, this conic program can be rewritten using a new matrix variable $Y = (y_{ij})_{i,j=0,\ldots,2n} = \left(\begin{array}{cc}1&x^\intercal\\x&X\end{array}\right)$ and symmetrizing as
\begin{align}
\nonumber \text{(P)}: \quad \min &\begin{pmatrix} 0 & 0 & 0\\ 0 & Q & 0\\ 0 & 0 & 0 \end{pmatrix} \bullet Y\\
\label{eq:one} \text{s.t.} &\begin{pmatrix} 1 & 0 & 0\\ 0 & 0 & 0\\ 0 & 0 & 0 \end{pmatrix}
\bullet Y = 1\\
\label{eq:two} &\begin{pmatrix} 0 & e_i^\intercal & e_i^\intercal\\
e_i & 0 & 0\\
e_i & 0 & 0
\end{pmatrix} \bullet Y = 2 &&\forall i = 1, \ldots, n\\
\label{eq:three} &\begin{pmatrix} 0 & 0 & 0 \\
0 & e_ie_i^\intercal & e_ie_i^\intercal\\
0 & e_ie_i^\intercal & e_ie_i^\intercal
\end{pmatrix} \bullet Y = 1 &&\forall i = 1, \ldots, n\\
\label{eq:four} &\begin{pmatrix} 0 & e_j^\intercal & 0 \\
e_j & -2e_je_j^\intercal &0\\
0 & 0 & 0
\end{pmatrix} \bullet Y = 0 &&\forall j = 1, \ldots, n\\
\nonumber &Y\in C^*_{1+n+n}
\end{align}

Since this holds for every possible choice of $Q = (q_{ij})$, we see that
\begin{align*}
\COR(n) = \{Z = (z_{ij})_{i, j = 1, \ldots, n} \mid \exists Y = (y_{ij})_{i, j = 0, \ldots, 2n} :\ &z_{ij} - y_{ij} = 0 &&\forall i, j = 1, \ldots, n\\ 
&\text{\eqref{eq:one}--\eqref{eq:four}}\\
& Y \in C^*_{1+2n}\}.
\end{align*}
Therefore, the correlation polytope $\COR(n)$ has a $\poly(n)$-size completely positive extension. 
\end{proof}

We point out that the completely positive extension of $\COR(n)$ constructed in the proof of the previous theorem is know to be not proper, that is, there is no $Y$ in the interior of $C^*_{1+2n}$ satisfying \eqref{eq:one}--\eqref{eq:four}. This was observed by Burer~\cite{Burer09}. Thus Theorem~\ref{thm:GPT} does not imply the existence of a cone factorisation for the slack matrix of $\COR(n)$ over the completely positive cone $\completelypositive_{1+2n}$. We now proceed to construct such a factorisation following another route.

First, we write down the dual of program (P) above:
\begin{align*}
\text{(D)}: \quad \sup \ &\alpha + \sum_{i=1}^n 2\beta_i + \sum_{i=1}^n \gamma_i\\
\text{s.t.}\ &\begin{pmatrix} 0 & 0 & 0\\ 0 & Q & 0\\ 0 & 0 & 0 \end{pmatrix} - 
\alpha \begin{pmatrix} 1 & 0 & 0\\ 0 & 0 & 0\\ 0 & 0 & 0 \end{pmatrix}
- \sum_{i=1}^n \beta_i \begin{pmatrix} 0 & e_i^\intercal & e_i^\intercal\\
e_i & 0 & 0\\
e_i & 0 & 0
\end{pmatrix}\\
& - \sum_{i=1}^n \gamma_i \begin{pmatrix} 0 & 0 & 0 \\
0 & e_ie_i^\intercal & e_ie_i^\intercal\\
0 & e_ie_i^\intercal & e_ie_i^\intercal
\end{pmatrix}
- \sum_{j=1}^n \delta_j \begin{pmatrix} 0 & e_j^\intercal & 0 \\
e_j & -2e_je_j^\intercal &0\\
0 & 0 & 0
\end{pmatrix} \in C_{1+n+n}.
\end{align*}

Call $M = M(\alpha,\beta,\gamma,\delta)$ the left-hand side of the conic constraint above. We claim that there is a choice of the variables of (D) such that $M$ is in the \emph{interior} of $C_{1+2n}$. In order to prove this, let $\beta_i = \delta_i = 0$ and $\gamma_i = \alpha$ for all $i = 1, \ldots, n$. Moreover, choose $\alpha < 0$ such that $\lambda - \frac{\alpha}{2} \geqslant 0$ for all eigenvalues $\lambda$ of $Q$. 

It is known that any symmetric matrix $M \in \RR^{(1+2n) \times (1+2n)}$ is in the interior of the copositive cone $C_{1+2n}$ if and only if ${\xi \choose x}^\intercal M {\xi \choose x} > 0$ for every ${\xi \choose x} \in \RR^{1+2n}_{\geqslant 0}$ that is not the all-zero vector \cite{Berman73,ds08}.

With our choice of $\alpha$, $\beta$, $\gamma$ and $\delta$, we have:
\begin{align*}
{\xi \choose x}^\intercal M {\xi \choose x} =
&\sum_{i = 1}^n \sum_{j = 1}^n q_{ij} x_i x_j
- \alpha \left\|{\xi \choose x}\right\|^2
- \sum_{i=1}^n 2 \alpha \, \xi (x_{i} + x_{i+n})\\
\geqslant &
\sum_{i = 1}^n \sum_{j = 1}^n (q_{ij} - \frac{\alpha}{2} \delta_{ij}) x_i x_j
- \frac{\alpha}{2} \left\|{\xi \choose x}\right\|^2\\
\geqslant &- \frac{\alpha}{2} \left\|{\xi \choose x}\right\|^2.
\end{align*}
By choosing $\alpha < 0$ with a large enough absolute value, we see that this quantity is always strictly positive when ${\xi \choose x} \in \RR^{1+2n}_{\geqslant 0}$ is different from the all-zero vector. So $M$ is in the interior of $C_{1+2n}$. This implies that Slater's condition is satisfied for (D), therefore strong duality holds (see, e.g., \cite{BV2004}). Since (P) is a reformulation of an optimisation problem over the correlation polytope, it is bounded and has a finite optimum value. Then by strong duality, (D) has an optimum value equal to the optimum value of (P). 

Later, we will need the fact that the optimum value of (D) is attained. Unfortunately, this is in general not implied by Slater's condition for (D).\footnote{This would be implied by Slater's condition for (P), but as noted before (P) is known to lack this property~\cite{Burer09}.} We prove it by using the specific structure of (D).

\begin{lem} \label{lem:dual_attained}
The optimum value of (D) is attained.
\end{lem}

\begin{proof}
Let $\kappa \leqslant 0$ denote the optimum value of (D). Fix any bound $B < \kappa$ and consider the set $\mathcal{F}_{\geqslant B}$ of feasible solutions of (D) of value at least $B$. It suffices to prove that $\mathcal{F}_{\geqslant B}$ is bounded. Indeed, $\mathcal{F}_{\geqslant B}$ is closed and by strong duality, we know that there exists a sequence of points of $\mathcal{F}_{\geqslant B}$ whose objective value converges to $\kappa$. If $\mathcal{F}_{\geqslant B}$ is bounded, then it is compact and the sequence admits a subsequence that converges in $\mathcal{F}_{\geqslant B}$. The limit of the subsequence is an optimal solution of (D).

Let $\alpha, \beta, \gamma, \delta$ be such $M = M(\alpha,\beta,\gamma,\delta)$ is in $\mathcal{F}_{\geqslant B}$. Because the objective value of $M$ is at least $B$, we have
\begin{equation}
\label{eq:near_opt}
\alpha + \sum_{i=1}^n 2 \beta_i + \sum_{i=1}^n \gamma_i \geqslant B
\end{equation}
Since $M$ is copositive, we have $z^\intercal M z \geqslant 0$ for every nonnegative vector $z$. In particular, we have 
\begin{align}
\label{eq:12} -\alpha &\geqslant 0\\
\label{eq:13} Q_{ii} -\gamma_i + 2 \delta_i &\geqslant 0 \qquad \text{for } i = 1, \ldots, n\\
\label{eq:14} -\gamma_i &\geqslant 0 \qquad \text{for } i = 1, \ldots, n\\
\label{eq:15} -\alpha - 2 \beta_i - \gamma_i &\geqslant 0 \qquad \text{for } i = 1, \ldots, n\\
\label{eq:16} -2 \alpha - 4\beta_i - \gamma_i &\geqslant 0 \qquad \text{for } i = 1, \ldots, n\\
\label{eq:17} -\alpha - \sum_{i=1}^n 2 \beta_i - \sum_{i=1}^n \gamma_i &\geqslant 0\\
\label{eq:18} -\alpha - 2 \sum_{i=1}^n 2 \beta_i - 2 \sum_{i=1}^n \gamma_i &\geqslant 0\\
\label{eq:19} -2\alpha -4\beta_i-4\delta_i+Q_{ii}-\gamma_i+2\delta_i &\geqslant 0 \qquad \text{for } i = 1, \ldots, n.
\end{align}

By \eqref{eq:near_opt} and \eqref{eq:18}, we have
$$
\frac{\alpha}{2} \geqslant B \underbrace{- \frac{\alpha}{2} - \sum_{i=1}^n 2 \beta_i - \sum_{i=1}^n \gamma_i}_{\geqslant 0} \geqslant B
$$
and thus, by \eqref{eq:12},
\begin{equation}
\label{eq:20}
2B \leqslant \alpha \leqslant 0.
\end{equation}

By \eqref{eq:near_opt} and \eqref{eq:12},
\begin{equation}
\label{eq:21}
\sum_{i=1}^n 2 \beta_i + \sum_{i=1}^n \gamma_i \geqslant B - \alpha \geqslant B
\end{equation}
%
%and
%$$
%\sum_{i=1}^n 2 \beta_i + \sum_{i=1}^n \gamma_i \leqslant -\alpha \leqslant -2B
%$$
%and
and by \eqref{eq:15} and \eqref{eq:20},
$$
2 \beta_i + \gamma_i \leqslant -\alpha \leqslant -2B
$$
for $i = 1, \ldots, n$. By \eqref{eq:21}, for $j = 1, \ldots, n$, 
$$
2 \beta_j + \gamma_j \geqslant B - \sum_{i \neq j} 2 \beta_i - \sum_{i \neq j} \gamma_i
\geqslant B + (n-1)2B = (2n-1) B.
$$
Summarizing, we have
\begin{equation}
\label{eq:22}
(2n-1) B \leqslant 2\beta_i+\gamma_i \leqslant -2B
\end{equation}
for $i = 1, \ldots, n$. 

Combining this with \eqref{eq:14}, we find
$$
2\beta_i \geqslant (2n-1)B
$$
for $i = 1, \ldots, n$. Now, using \eqref{eq:16}, \eqref{eq:20} and \eqref{eq:22},
$$
2\beta_i \leqslant -2\alpha-2\beta_i-\gamma_i \leqslant -4B - (2n-1)B = -(2n+3)B
$$
and thus, by \eqref{eq:22} again,
$$
\gamma_i \geqslant (2n-1)B - 2\beta_i \geqslant (2n-1)B + (2n+3)B = (4n+2)B
$$
for $i = 1, \ldots, n$.

Thus we find the following bounds:
$$
(2n-1)B \leqslant 2\beta_i \leqslant -(2n+3)B
\qquad
\text{and}
\qquad
(4n+2)B \leqslant \gamma_i \leqslant 0
$$
for $i = 1, \ldots, n$.

It is routine to obtain upper and lower bounds on $\delta_i$ for each $i = 1, \ldots, n$ from the bounds on the other variables, using \eqref{eq:13} and \eqref{eq:19}. We obtain
$$
(4n+2)B - Q_{ii} \leqslant 2\delta_i \leqslant Q_{ii} -(8n+2) B.
$$
Thus $\mathcal{F}_{\geqslant B}$ is bounded, because it is contained in a box.
\end{proof}

We will now use this to give a completely positive factorisation for the slack matrix of the correlation polytope $\COR(n)$. Let $S(n)$ denote the slack matrix of the correlation polytope $\COR(n)$, with respect to its vertices and facet-defining inequalities. For each given vertex $aa^\intercal$ of $\COR(n)$, Bob creates a completely positive matrix
$$
Y=Y(a)=
\begin{pmatrix} 1\\ a\\ \mathbf{1}-a \end{pmatrix}
\begin{pmatrix} 1\\ a\\ \mathbf{1}-a \end{pmatrix}^\intercal
$$
which clearly lies in the feasible region of (P). For a given facet-defining inequality $\sum_{i=1}^n \sum_{j=1}^n q_{ij} y_{ij} \geqslant \kappa$, Alice uses the coefficients $\alpha$, $\beta$, $\gamma$, and $\delta$ from an optimal solution of (D) ---which exists by Lemma~\ref{lem:dual_attained}--- to define a copositive matrix $M = M(\alpha,\beta,\gamma,\delta)$ for the given facet. We then have
\begin{align*}
Y \bullet M 
&= Y \bullet \left( 
\begin{pmatrix} 0 & 0 & 0\\ 0 & Q & 0\\ 0 & 0 & 0 \end{pmatrix}
- \alpha \begin{pmatrix} 1 & 0 & 0\\ 0 & 0 & 0\\ 0 & 0 & 0 \end{pmatrix}
- \sum_{i=1}^n \beta_i \begin{pmatrix} 0 & e_i^\intercal & e_i^\intercal\\
e_i & 0 & 0\\
e_i & 0 & 0
\end{pmatrix} \right.\\
&\qquad \qquad \left.
- \sum_{i=1}^n \gamma_i \begin{pmatrix} 0 & 0 & 0 \\
0 & e_ie_i^\intercal & e_ie_i^\intercal\\
0 & e_ie_i^\intercal & e_ie_i^\intercal
\end{pmatrix}
- \sum_{j=1}^n \delta_j \begin{pmatrix} 0 & e_j^\intercal & 0 \\
e_j & -2e_je_j^\intercal &0\\
0 & 0 & 0
\end{pmatrix} 
\right)\\
&= Q \bullet aa^\intercal - \left( \alpha + \sum_{i=1}^n 2\beta_i + \sum_{i=1}^n \gamma_i \right)\\ 
&= Q \bullet aa^\intercal - \kappa            
\end{align*}
where the middle equality follows from the fact that $Y$ is feasible for (P) and the last equality follows from the optimality of the dual solution given by $\alpha$, $\beta$, $\gamma$, and $\delta$. Thus $Y \bullet M$ is the entry of the slack matrix $S = S(n)$ of $\COR(n)$ corresponding to the given facet $\sum_{i=1}^n \sum_{j=1}^n q_{ij} y_{ij} \geqslant \kappa$ and given vertex $aa^\intercal$. The different matrices $Y$ (one for each vertex) and $M$ (one for each facet) give us a factorisation of the slack matrix $S(n)$ of the correlation polytope over the completely positive cone $\completelypositive_{1+2n}$. 

\subsection{Polynomially definable $0/1$-polytopes}

A polytope $P \subseteq \RR^d$ is called a \emph{$0/1$-polytope} if all its vertices are in $\{0,1\}^d$. Now fix a polynomial $p = p(d)$. Informally, we say that a $0/1$-polytope $P$ is ``$p(d)$-definable'' if there exists a predicate defining the vertex set of $P$ that is efficient in the sense that it can be implemented by a circuit of ``size'' at most $p(d)$. Formally, $0/1$-polytope $P \subseteq \RR^d$ is said to be \emph{$p(d)$-definable} if there exists a Boolean circuit $C(x,y)$ with $k+d$ inputs $x \in \{0,1\}^k$, $y \in \{0,1\}^d$, one output and at most $p(d)$ gates, such that
$$
P = P(C,x) = \conv \{y \in \{0,1\}^d \mid C(x,y) = 1\}.
$$
The idea is that the circuit $C(x,y)$ checks whether $y \in \{0,1\}^d$ is a vertex of $P$ or not, given \emph{advice} $x \in \{0,1\}^k$. The bits of $x$ give side information that is used to define the vertex set of $P$. 

For instance, the stable set polytope of a $n$-vertex graph $G$ is $2n^2$-definable because there exists a circuit $C(x,y)$ with $k = {n \choose 2}$ plus $d = n$ inputs, one output and $3 {n \choose 2}+1 \leqslant 2n^2$ gates ---$2 {n \choose 2}+1$ AND gates and ${n \choose 2}$ NOT gates--- that checks whether $y$ is the incidence vector of a stable set\footnote{Recall that a set $S$ of vertices of $G$ is said to be \emph{stable} (or \emph{independent} if no two vertices of $S$ are linked by an edge of $G$.} $S$ in $G$, given the incidence vector $x$ of the edge set of $n$-vertex graph $G$. Thus for $x = \langle G \rangle$ and $y = \langle S \rangle$, the circuit $C$ simply checks whether for each edge $ij$ of $G$ we have $i \notin S$ or $j \notin S$. 

In some cases, the advice bits are not necessary and we can let $k = 0$. For instance, for the case of the correlation polytope, one can easily design a circuit $C(y)$ for $d = n^2$ with $O(d) = O(n^2)$ gates that tests whether $y \in \{0,1\}^d$ is a binary correlation vector $y = bb^\intercal$ for some $b \in \{0,1\}^n$.

Often one does not consider single polytopes but families of $0/1$-polytopes defined in various dimensions. An example of such a family is that of all correlation polytopes, that is, $\{\COR(n) \mid n \geqslant 1\}$. We say that such a family of $0/1$-polytopes is \emph{polynomially definable} if all its members that live in $\RR^d$ are $p(d)$-definable, for the same polynomial $p(d)$. Examples of polynomially definable families of polytopes include the correlation polytopes (or cut polytopes, since they are linearly equivalent to the correlation polytopes), and many others such as travelling salesman polytopes, stable set polytopes, and so on. As a matter of fact, most commonly studied families of $0/1$-polytopes are polynomially definable, since they have the property that recognising vertices can be done efficiently.

Here, we slightly strengthen an interesting observation of Maximenko~\cite{Maximenko12} about correlation polytopes. He proves that every $p(d)$-definable $0/1$-polytope can be obtained as the projection of some face of some correlation polytope $\COR(n)$ with $n$ polynomial in $d$. In this sense, correlation polytopes are ``universal objects''. In fact, Maximenko restricts himself to the case where the predicate defining the $0/1$-polytope has no side input $x$, that is, $k = 0$. (He also uses the cut polytopes, which are linearly equivalent to the correlation polytopes, and this makes the proof a bit more complicated). We give a short proof of the unabridged version of his result, see Theorem~\ref{thm:COR-simulation}.

In order to make the proof of Theorem \ref{thm:COR-simulation} as short and transparent as possible, we assume here that the circuit $C$ is implemented using only NOR gates (with fan-in $2$ and unbounded fan-out). Recall that $\NOR(z_i,z_j) = 1$ if and only if $z_i = z_j = 0$. It is known that any circuit using standard gates (OR, AND, NOT, \ldots) can be transformed into such a circuit, with only a polynomial blow-up in size (constant blow-up for circuits with bounded fan-in). In particular, $\NOT(z_i) = \NOR(z_i,z_i)$ and we allow pairs of parallel arcs in the circuit to be able to repeat inputs.

Finally, remark that, although we do not put any explicit bound on the size of advice $x$, it is clear that the circuit $C(x,y)$ can only read at most $2p(d)$ many bits of $x$ since it has at most $p(d)$ gates.

\begin{thm}
\label{thm:COR-simulation}
Every $p(d)$-definable $0/1$-polytope $P \subseteq \RR^d$ is a projection of a face of the correlation polytope $\COR(n)$, with $n \leq d + p(d)$.
\end{thm}

\begin{proof}
Suppose $P = P(C,x)$ for some circuit $C(x,y)$ with at most $p(d)$ gates, and some $x \in \{0,1\}^k$. Assume that $P$ is not empty, otherwise the result is trivial. We create $n \leq d + p(d)$ new Boolean variables which we denote as $z_{i,i}$ for $i = 1, \ldots, n$. The first $d$ variables $z_{i,i}$ represent the variables $y_i$, and the last at most $p(d)$ variables represent the output values of the NOR gates. With a slight abuse of notations, we also denote the variables of the space in which $\COR(n)$ is defined as $z_{i,j}$, where $i, j \in \{1,\ldots,n\}$. Notice that none of the variables $z_{i,j}$ corresponds to the advice bits $x_\ell$, which are considered as constants on which the polytope $P$ depends.

For defining a face $F$ of $\COR(n)$ that projects to $P$, we intersect $\COR(n)$ with  $p(d) + 1$ valid hyperplanes, one for each NOR gate and one for the output of the circuit. 

For each NOR gate $z_{k,k} = \NOR(z_{i,i},z_{j,j})$ we add the equation 
\begin{equation}
\label{eq:NOR}
z_{i,i} + z_{j,j} - z_{i,j} + z_{k,k} - 2 z_{i,k} - 2 z_{j,k} = 1\ .
\end{equation}
This equation defines an hyperplane that is valid\footnote{A hyperplane is said to be \emph{valid} for a polytope $P$ if $P$ lies completely (not necessarily strictly) on one side of the hyperplane.} for $\COR(n)$ because the left-hand side is always less or equal to $1$ (recall for instance that $z_{i,i} = 1$ and $z_{j,j} = 1$ implies $z_{i,j} = 1$), and equals $1$ for a vertex of $\COR(n)$ if and only if $(z_{i,i},z_{j,j},z_{k,k}) \in \{(1,0,0),(0,1,0), (1,1,0), (0,0,1)\}$, i.e. if and only if $z_{k,k} = \NOR(z_{i,i},z_{j,j})$. 

For a NOR gate involving one constant (e.g., one of the advice bits), that is, a NOR gate of the form $z_{k,k} = \NOR(z_{i,i},c)$, we have $z_{k,k} = \NOT(z_{i,i})$ if $c = 0$ and $z_{k,k} = 0$ if $c = 1$. Gates of this type can also be easily simulated by valid hyperplanes similar to \eqref{eq:NOR}. For instance we can use:
\begin{equation}
\label{eq:NORc0}
z_{i,i} + z_{k,k} - 2 z_{i,k} = 1\ (\mbox{when }c=0)\ ;
\end{equation}
\begin{equation}
\label{eq:NORc1}
z_{k,k} = 0\ (\mbox{when }c=1)\ .
\end{equation}

The output of a NOR gate involving two constants (e.g., two advice bits) is simply considered as a new constant whose value is a function of $x \in \{0,1\}^k$.
 
Finally, assuming that $z_{n,n}$ represents the output of the circuit, we add the equation
\begin{equation}
\label{eq:output}
z_{n,n} = 1
\end{equation}
that defines a valid hyperplane.

The face $F$ is thus the intersection of $\COR(n)$ and the hyperplanes eqs. (\ref{eq:NOR},\ref{eq:NORc0},\ref{eq:NORc1},\ref{eq:output}).
From the above construction it follows that $y\in P$ (that is $C(x,y) = 1$) if and only if $y_i=z_{i,i}$ for $i = 1, \ldots, d$ and $z\in F$. Therefore, the image of $F$ by the projection to the variables $z_{i,i}$ for $i \in \{1,\ldots,d\}$ is exactly $P$.
\end{proof}

\noindent
By combining this with Theorem \ref{thm:COR_poly}, we obtain the following result.

\begin{thm}
Every polynomially definable $0/1$-polytope has a polynomial size completely positive extension.
\end{thm}

This proves that virtually all problems of interest in combinatorial optimisation can be expressed in an economical way as conic programs over the completely positive cone, which is striking given the large number of papers establishing this for individual problems, e.g.~\cite{QKRT98,BDKRT2000,KP02,Burer09}. Here, we consider a combinatorial optimisation problem as the task of finding in a (possibly implicitly described) collection $\mathcal{F}$ of subsets of a finite universe $U$ whose elements have weights $w(u) \in \mathbb{R}$, a set $F \in \mathcal{F}$ such that $w(F) = \sum_{u \in F} w(u)$ is maximum (minimum), which generically corresponds to maximizing (minimizing) a linear function over a $0/1$ polytope. (More general would be to optimize over integer polyhedra, that is, polyhedra whose vertices have integer coordinates; this is part of discrete optimisation.) In terms of GPTs, the theorem supplies a large number of communication problems that are easy for GPTs based on completely positive / copositive cones, but hard in the classical case.

\section*{Acknoweledgments}

We acknowledge financial support by european projects QCS and QALGO, by \emph{Fonds de la Recherche Scientifique - FNRS}, and by the \emph{Actions de Recherche Concert\'ees} (ARC) fund of the \emph{Communaut\'e fran\c{c}aise de Belgique}. Finally, we thank our colleague Sebastian Pokutta for discussions about this paper.

\end{document}